\documentclass[pra,twocolumn, showpacs, superscriptaddress
]{revtex4-1}
\usepackage{graphicx}
\usepackage{amsfonts}
\usepackage{amssymb}
\usepackage{amsthm}
\usepackage{array}
\usepackage{amsmath}
\usepackage{verbatim} 
\usepackage{hyperref}
\usepackage{color}
\usepackage{bbold}
\usepackage{epstopdf}
\usepackage{mathtools}
\usepackage{enumerate}

\newtheorem{proposition}{Proposition}
\newtheorem*{proposition*}{Proposition}
\newtheorem{definition}{Definition}
\newtheorem{theorem}{Theorem}
\newtheorem*{theorem*}{Theorem}

\newtheorem*{corollary*}{Corollary}

\newcommand{\ket}[1]{\left\vert#1\right\rangle}
\newcommand{\bra}[1]{\left\langle#1\right\vert}

\newcommand{\abs}[1]{\left|#1\right|}
\newcommand{\ketbra}[2]{\vert #1 \rangle\langle #2 \vert}
\newcommand{\braket}[2]{\langle #1 | #2 \rangle}

\def\bra#1{\langle #1|}
\def\ket#1{\left|#1 \right>}

\def\Tr{\mbox{Tr}}

\begin{document}
\title{Quantifying the Coherence Between Coherent States}

\author{Kok Chuan Tan}
\email{bbtankc@gmail.com}
\affiliation{Center for Macroscopic Quantum Control, Department of Physics and Astronomy, Seoul National University, Seoul, 151-742, Korea}
\author{Tyler Volkoff}
\affiliation{Department of Physics, Konkuk University, Seoul 05029, Korea}
\author{Hyukjoon Kwon}
\affiliation{Center for Macroscopic Quantum Control, Department of Physics and Astronomy, Seoul National University, Seoul, 151-742, Korea}
\author{Hyunseok Jeong}
\email{h.jeong37@gmail.com}
\affiliation{Center for Macroscopic Quantum Control, Department of Physics and Astronomy, Seoul National University, Seoul, 151-742, Korea}

\date{\today}

\begin{abstract}
In this paper, we detail an orthogonalization procedure that allows for the quantification of the amount of coherence present an arbitrary superposition of coherent states.
The present construction is based on the quantum coherence resource theory introduced by Baumgratz {\it et al.}~\cite{Baumgratz14}, and the coherence resource monotone that we identify is found to characterize the nonclassicality traditionally analyzed via the Glauber-Sudarshan $P$ distribution. This suggests that identical quantum resources underlie both quantum coherence in the discrete finite dimensional case and the nonclassicality of quantum light. We show that our construction belongs to a family of resource monotones within the framework of a resource theory of linear optics, thus establishing deeper connections between the class of incoherent operations in the finite dimensional regime and linear optical operations in the continuous variable regime.
\end{abstract}


\maketitle

\section{Introduction} 

The remarkable differences between the classical theories and quantum theories have long been a captivating and fruitful area of study for physicists, culminating in recent decades in the realization that such differences can used to perform a variety of useful informational tasks~\cite{NielsenChuang}. Subsequent developments have led to the identification and quantification of nonclassical quantum properties such as quantum entanglement~\cite{Horodecki2001}, nonlocality~\cite{Brunner2014} and quantum discord~\cite{Ollivier2001}. These remain intense areas of research, with new applications that exploit these nonclassical properties still being found~\cite{Lloyd2008,Chuan2013, Dakic2012, Girolami2014, Acin2016}.

A fairly recent development in the quantum resources arena is the introduction of a resource theory of quantum coherence by Baumgratz {\it et al.}~\cite{Baumgratz14}. This resource theory draws its primary inspiration from a similar program that was successful in the study of entanglement~\cite{Plenio07}. Adopting this approach for quantum coherence necessarily requires the assumption of some natural basis which is assumed to form an orthonormal set $\{ \ket{i} \}$ where states $\ket{i}$ are typically considered to be ``classical". Since the initial proposal by Baumgratz {\it et al.}~\cite{Baumgratz14}, other variations of such resource theories have also been explored~\cite{Chitambar2017}. For our purposes, we will limit our attention to the version originally proposed in~\cite{Baumgratz14}. Recent literature have applied this resource theory to the study of a diverse range of topics, such as quantum correlations~\cite{Tan2016, Streltsov15}, interferometric experiments~\cite{Wang2017} and quantum estimation~\cite{Giorda2016}.

The aforementioned resource theory of coherence typically considers finite dimensional quantum systems. At the opposite end of the spectrum, we may also consider its application in the infinite dimensional (or continuous variable) regime, of which quantum systems of light is a prime example. Recent attempts to quantify the coherence in such systems have mainly focused their attention on the Fock basis $\{ \ket{n} \}$, assuming that Fock diagonal states are the free ``classical" resource in the infinite dimensional regime~\cite{Zhang2015,Xu2016}.

This approach is however, in diametric opposition to the traditional notions of classical light based on the Glauber-Sudarshan P representation of the state of the electromagnetic field. Indeed, Fock states are decidedly nonclassical~\cite{Sperling2016, Ryl2017}. The most general notion of classical light have already been extensively studied since the 1960s~\cite{Glauber1963a,Glauber1963b,Glauber1963c,Klauder1968}, and it is well established that 
the quantum states of light that most closely resemble classical light fields, both in photon statistics and dynamics, are the so-called coherent states~\cite{Glauber1965}. It is therefore desirable that any quantification of quantum coherence for quantum states of light involves the set of coherent states. Unfortunately, the set of coherent states is overcomplete; in particular, the coherent states do not form a mutually orthonormal set, and therefore do not permit the direct application of the resource theoretical approach outlined in~\cite{Baumgratz14}. 

The question is then the following: suppose we would like to consider the coherence with respect to the set of coherent states $\{ \ket \alpha \}$, how do we quantify this and in what sense is it nonclassical? In this paper, we propose a resolution to this problem. In so doing, we will demonstrate that the quantum resource identified by Baumgratz {\it et al.}~\cite{Baumgratz14} is essentially the same as the notion of nonclassicality identified by Glauber~\cite{Glauber1965}. We will also demonstrate that this nonclassical resource is closely related to what we refer to as a resource theory of linear optics.


\section{Preliminaries}

We will adopt the axiomatic approach for coherence measures as shown in Ref.~\cite{Baumgratz14}. The essential ingredients are as follows.

For a fixed basis $\{ \ket{i} \}$, the set of incoherent states $\cal I$ is the set of quantum states with diagonal density matrices with respect to this basis. Given this, we say that  $\mathcal{C}$ is a measure of quantum coherence if it satisfies following properties:
(C1) $C(\rho) \geq 0$ for any quantum state $\rho$ and equality holds if and only if $\rho \in \cal I$.
(C2a) The measure is non-increasing under incoherent completely positive and trace preserving maps (ICPTP) $\Phi$ , i.e., $C(\rho) \geq C(\Phi(\rho))$.
(C2b) Monotonicity for average coherence under selective outcomes of ICPTP:
$C(\rho) \geq \sum_n p_n C(\rho_n)$, where $\rho_n = \hat{K}_n \rho \hat{K}_n^\dagger/p_n$ and $p_n = \Tr [\hat{K}_n \rho \hat{K}^\dagger_n ]$ for all $\hat{K}_n$ with $\sum_n \hat{K}_n \hat{K}^\dagger_n = \mathbb 1$ and $\hat{K}_n {\cal I} \hat{K}_n^\dagger \subseteq \cal I$.
(C3) Convexity, i.e. $\lambda C(\rho) + (1-\lambda) C(\sigma) \geq C(\lambda \rho + (1-\lambda) \sigma)$, for any density matrix $\rho$ and $\sigma$ with $0\leq \lambda \leq 1$.

We will also frequently make reference to the set of coherent states which we denote by $\{ \ket{\alpha} \}$. (For an overview, see for instance~\cite{Adesso2014}). It is known that every quantum state of light $\rho$ permits a representation that is diagonal with respect to coherent states, i.e.

$$
\rho = \int d^2\alpha P(\alpha)\ket{\alpha}\bra{\alpha}
$$

where the coefficient $P(\alpha)$ is called the Glauber-Sudarshan $P$ distribution~\cite{Titulaer1965}. The P distribution always sums to 1 but may display negativities, in which case it is considered nonclassical. On the other hand, $P$ distributions that exhibit the properties of a classical, nonnegative probability distribution are considered to have classical analogues.

Finally, we will also make references to \textit{linear optical operations}, which will require some clarification. Here, we specifically take this term to refer to the set of
passive unitary optical operations that can be performed using basic building blocks of beam splitters, phase shifters, half and quarter wave plates as described in Ref~\cite{Kok2007} supplemented with displacement operations, defined by $D(\alpha) \coloneqq e^{(\alpha a^\dag - \alpha^* a)}$. In contrast, the most general linear transformation of the Bogoliubov type includes operations such as squeezing operations, that can give rise to highly nonclassical light. In our context, the defining property of such a linear optical operation is that if 
the input quantum state is given by pure, classical light of the form $\ket{\vec{\alpha}} = \ket{\alpha^1} \ldots \ket{\alpha^k}$, then the output state is also pure and classical, i.e., if $U$ is a unitary linear optical operation, then $U\ket{\vec{\alpha}} = \ket{\vec{\beta}} =  \ket{\beta^1} \ldots \ket{\beta^k}$.

\section{Example for pure states}

The key idea that we will present here is to preprocess a general quantum state using an orthogonalization procedure which shares some superficial similarities with the Gram-Schmidt orthogonalization procedure from linear algebra. We first illustrate the process using some orthogonal basis states, and show that this procedure, when applied to an orthogonal basis, can be interpreted as a generalization of the concept of coherence proposed by Baumgratz \textit{et al.}~\cite{Baumgratz14}.

Consider some orthogonal basis $\{ \ket{i} \}$ with $i = 1, \ldots, N$ in some $N$ dimensional Hilbert space and some arbitrary quantum state $\ket{\psi} = \sum_i c_i \ket{i}$. Without any loss in generality, we assume that the coefficients are in decreasing order, so $\abs{c_i}\geq \abs{c_{i+1}}$. We now describe a preprocessing procedure. We define a CNOT type operation performing the operation $U_i \ket{i}\ket{0} = \ket{i}\ket{i}$. 

Suppose we perform a series of such CNOT type operations starting from the basis state with the largest overlap with $\ket{\psi}$, so $U = U_N \ldots U_1$. Applying this unitary, the final result is the state $U\ket{\psi} = \sum_i c_i \ket{i}\ket{i}$. We note that the coherence with of $U\ket{\psi}$ in the basis $\{\ket{i}\ket{i}\}$ is the same as the coherence of  $\ket{\psi}$ in the basis $\{\ket{i}\}$. Therefore, from the perspective of coherence, the preprocessing procedure is completely superfluous, and the coherence before and after the process is completely described by the same coefficients $c_i$.

In the above example, the unitary procedure turned out to be extraneous because the initial basis states are chosen to be orthonormal. However, when the initial reference set of states is not an orthonormal set, such as when the set of states considered are the coherent states $\{ \ket \alpha \}$, then we see that the operation may not be trivial. Suppose we have some initial pure state $|\psi\rangle$. If one were to similarly define a series of CNOT type operations as before, with the exception that the control states are drawn from the non-orthonormal set $\{ \ket \alpha \}$, we see that the resulting state will have the form $U\ket{\psi}\ket{0} = \sum_i c'_i \ket{\alpha'_i}\ket{\beta'_i}$, where the set of states $\{ \ket{\alpha'_i}\ket{\beta'_i} \}$ will be orthonormal so long as $\braket{\beta_i}{\beta_j} = \delta_{ij}$. We note that this orthogonality condition can always be strictly enforced by an encoding across multiple spatial/polarization modes, but for notational simplicity, we will instead use some set of sufficiently well separated coherent states within a single mode, $\{ \ket{\beta_i} \}$ , which can be chosen to be arbitrarily close to orthonormal.

We now describe the orthogonalization procedure with respect to the set of coherent states in detail. Following the same argument as above, let us define $\ket{\psi_i}$ through the recursion relation $\ket{\psi_i} = \ket{\psi_{i-1}} -  \ket{\alpha_{i-1}}\braket{\alpha_{i-1}}{\psi_{i-1}}$, where the coherent state $\ket{\alpha_i}$ satisfies $\braket{\alpha_i}{\psi_i} = \max_{\alpha'}\braket{\alpha'}{\psi_i}$ and the initial state $\ket{\psi_1} = \ket{\psi}$ is some given pure quantum state of interest. 

Given some finite series of vectors $\{ \ket{\alpha_i} \}$ where $i = 1,\ldots,N$, let the CNOT type unitary be defined to be $U_{\alpha_i} \coloneqq \ketbra{\alpha_i}{\alpha_i} \otimes \ketbra{\beta_i}{0} + \ketbra{\alpha_i}{\alpha_i} \otimes \ketbra{0}{\beta_i} + (\openone \otimes \openone - \ketbra{\alpha_i}{\alpha_i}\otimes \ketbra{0}{0} - \ketbra{\alpha_i}{\alpha_i}\otimes \ketbra{\beta_i}{\beta_i} )$. This definition essentially performs the map $U_{\alpha_i}\ket{\alpha_i}\ket{0} =\ket{\alpha_i}\ket{\beta_i}$. From this, we can construct the unitary map just as before: $U_{GS} = U_{\alpha_{N}} \ldots U_{\alpha_{1}}$.

We will call $U_{GS}$ the Gram-Schmidt unitary, since it performs an orthogonalization process. The end result is some orthogonal subspace spanned by $\{ \ket{\alpha_i}\ket{\beta_i} \}$ where $i=1,\ldots, N$. Within this $N$ dimensional subspace, the discrete finite dimensional formulation of coherence will then apply. 

\section{Generalization to mixed states.}

The following is a construction of $U_{GS}$ which will appropriately generalize the definition for mixed states:

\begin{definition} [Gram-Schmidt Unitary] \label{def::GSUnitary}
For a given density matrix $\rho_A$, let $\rho^{(0)}_{AB} = \rho_A \otimes \ket{0}_B\bra{0}$. 

We now define $\ket{\alpha^{(i)}}$ to be some coherent state  achieving the optimal value $\mathrm{Tr}(\ket{\alpha^{(i)}}\bra{\alpha^{(i)}} \otimes \ket{0}_B\bra{0}\rho^{(i)}_{AB}) = \max_\alpha \mathrm{Tr}(\ket{\alpha}_A\bra{\alpha} \otimes \ket{0}_B\bra{0}\rho^{(i)}_{AB})$, where $\rho^{(i)} \coloneqq U_{\alpha^{(i-1)}}\rho^{(i-1)}U^\dag_{\alpha^{(i-1)}}$ and $U_{\alpha_i} \coloneqq \ketbra{\alpha_i}{\alpha_i} \otimes \ketbra{\beta_i}{0} + \ketbra{\alpha_i}{\alpha_i} \otimes \ketbra{0}{\beta_i} + (\openone \otimes \openone - \ketbra{\alpha_i}{\alpha_i}\otimes \ketbra{0}{0} - \ketbra{\alpha_i}{\alpha_i}\otimes \ketbra{\beta_i}{\beta_i})$ is a CNOT type unitary. We assume that $\{\ket{0}, \ket{\beta_i} \}$ forms some set of mutually orthonormal vectors.

Let $N > 1$ be some integer. Then the following unitary:

$$
U^{(N)}_{GS} = U_{\alpha^{(N)}} \ldots U_{\alpha^{(0)}}
$$

is called the $N$th Gram-Schmidt unitary. Note that in general, $U^{(N)}_{GS}$ depends on the state $\rho_A$.
\end{definition}

In the case of degeneracy, where more than one coherent state may achieve $\max_\alpha \mathrm{Tr}(\ket{\alpha}_A\bra{\alpha} \otimes \ket{0}_B\bra{0}\rho^{(i)}_{AB})$, the choice of unitaries above is not necessarily unique. To accommodate this, we will also define the set of all possible choices of such unitaries  $\mathcal{S}^{(N)}$. We can also generalize to the case of multimode states by considering $\ket{\vec{\alpha_i}} \coloneqq \ket{\alpha_i^1}\ldots \ket{{\alpha_i^k}}$ in place of $\ket{\alpha_i}$, so that our treatment here can be made as general as possible.

After the orthogonalization process, a pure state will have the form $U\ket{\psi}\ket{0} = c_0\ket{\epsilon}\ket{0}+\sum_{i=1}^N c_i \ket{\alpha_i}\ket{\beta_i}$, where the set of states $\{ \ket{\alpha_i}\ket{\beta_i} \}$ will be orthogonal. The vector $\ket{\epsilon}\ket{0}$ represents the portion of the vector space that is not orthogonalized by the $N$th Gram-Schmidt unitary., which we can always remove by projecting onto the subspace spanned by  $\{ \ket{\alpha_i}\ket{\beta_i} \}$. We introduce the following quantity:

\begin{definition} [$N$-coherence]
For some discrete finite dimensional coherence measure $\mathcal{C}$, we define the $N$-coherence $\mathcal{C}_{\alpha}$ for a pure state $\ket{\psi}$ to be:

$$
\mathcal{C}_{\alpha}(\ket{\psi}, N) = \inf_{U^{(N)}_{GS}\in \mathcal{S}^{(N)}} \mathcal{C}[\Phi^{(N)}_{GS}(\ketbra{\psi}{\psi}) ]
$$

where $\Phi^{(N)}_{GS}(\rho) =  \Pi^{(N)}_{GS}(U^{(N)}_{GS} (\rho\otimes \ketbra{0}{0})U^{(N)\dag}_{GS})\Pi^{(N)}_{GS} / \mathrm{Tr}(  \Pi^{(N)}_{GS}(U^{(N)}_{GS} (\rho\otimes \ketbra{0}{0})U^{(N)\dag}_{GS})\Pi^{(N)}_{GS})$ is called the $N$th Gram-Schmidt map. The projector $\Pi^{(N)}_{GS} \coloneqq \sum_{i=1}^N \ketbra{\alpha_i}{\alpha_i} \otimes \ketbra{\beta_i}{\beta_i} $ is the projection onto the $N$ dimensional subspace spanned by $\{ \ket{\alpha_i} \ket{\beta_i} \}$ where the vectors $\{ \ket{\alpha_i} \}$ and $\{\beta_i\} $ are the same vectors previously defined in Definition~\ref{def::GSUnitary}.  More generally, for any mixed quantum state $\rho_A$, we employ the following definition:

$$
\mathcal{C}_{\alpha}(\rho_A,N) = \inf_{(\rho_{AE}, U_{GS})\in (\mathcal{E},\mathcal{S}^{(N)})} \mathcal{C}[\Phi^{(N)}_{GS}(\rho_{AE})]
$$

where  $\mathcal{E} \coloneqq \{ \rho_{AE} \mid \mathrm{Tr}{\rho_{AE}} = \rho_A \}$ is the set of extensions of $\rho_A$. The coherence $\mathcal{C}$ is measured with respect to the set of orthogonal vectors $\{\ket{\vec{\alpha}_i}\ket{\beta_i}\}$ specified by $U^{(N)}_{GS}$.
\end{definition}

In general, we allow the the coherence measure $\mathcal{C}$ to be any finite dimensional coherence measure satisfying the axioms listed in Ref.~\cite{Baumgratz14}, with only one additional requirement. The coherence measure $\mathcal{C}$ should be asymptotically continuous in the sense that if some state $\rho$ has infinitesimally small coherence, then it is infinitesimally close to some incoherent state $\sigma$. That is, if we have some sequence of states $\rho^n$ such that $\lim_{n \rightarrow \infty}\mathcal{C}(\rho^n) = 0$ then for every $\epsilon > 0$, there is some $n_{max}$ such that for every $n>n_{max}$, there exists some incoherent state $\sigma^n$ such that $\frac{1}{2}\Vert \rho^n - \sigma^n \Vert_{tr} < \epsilon$. This is satisfied, for instance, by both coherence measures introduced in Ref.~\cite{Baumgratz14}. This is because both the $l_1$ norm~\cite{Deza2009} and the relative entropy~\cite{Audenaert2005} are lower bounded by the trace norm. 


Next, we define the $\epsilon$ smoothed version of the above quantity so as to consider states in the immediate vicinity of the state of interest.

\begin{definition} [$\epsilon$ smoothed $N$-Coherence]

The $\epsilon$-smoothed $N$-Coherence for some $\epsilon >0$ is the quantity:

$$\mathcal{C}_{\alpha}(\rho_A,N, \epsilon) \coloneqq \inf_{\rho'_A \in \mathcal{B}(\rho_A, \epsilon)} C(\rho'_A,N) $$

where $\mathcal{B}(\rho_A,\epsilon) = \{ \rho_A' \mid \frac{1}{2} \Vert \rho_A' - \rho_A \Vert_{tr} \leq \epsilon \}$ is the $\epsilon$ ball centred at $\rho_A$ with respect to the trace norm.
\end{definition}

Finally, the main figure of merit that we consider is the following:

\begin{definition} [$\alpha$-coherence] \label{def::alphaCoherence}

The $\alpha$-coherence is the limiting value of the $\epsilon$ smoothed $N$-Coherence:

$$
\mathcal{C}_{\alpha}(\rho_A) \coloneqq \lim_{\epsilon \rightarrow 0}\lim _{N \rightarrow \infty}\mathcal{C}_{\alpha}(\rho_A, N, \epsilon) .
$$

\end{definition}

In Definition~\ref{def::alphaCoherence}, we have combined the finite dimensional formulation of coherence with that of non-classical systems of light. The $\alpha$-coherence may therefore be interpreted as the limiting case of the coherence identified by Baumgratz \textit{et al.}~\cite{Baumgratz14}, optimized over state extensions and all degenerate cases, if any. Coherence effects are typically signs of non-classicality if an appropriate basis is chosen. It remains to be shown what kind of non-classicality the above quantity measures. In the following section, we will argue that the $\alpha$-coherence is closely related to non-classicality in the sense of negative Glauber-Sudarshan P distributions.

\section{Main Results}

Here, we present the main properties of the $\alpha$-coherence and 
its relation to traditional notions of coherence in the quantum optics literature. We first prove that, for a given state $\rho_{A}$, a vanishing value of the $\alpha$-coherence is equivalent to the existence of a Glauber-Sudarshan P distribution (referred to hereafter simply as the P distribution) for $\rho_{A}$ which is a probability density on the complex plane. A nonzero value of the $\alpha$-coherence is, therefore, an indicator of non-classicality.

\begin{theorem}

The $\alpha$-coherence $\mathcal{C}_\alpha(\rho_A) = 0$ iff $\rho_{A}$ is a classical state.

\end{theorem}

\begin{proof}
Let $\rho_{A}$ have P distribution $P_{\rho_{A}}(\alpha)$ which is the density of a regular Borel probability measure on the complex plane. By the density (in the weak-* topology) of the Dirac point measures on the space of regular Borel measures on $\mathbb{C}$, it follows that given $\epsilon >0$ and a continuous function $f(\alpha)$ on $\mathbb{C}$ that vanishes at infinity, there exists a finite linear combination of point measures $\sum_{j=0}^{m}c_{j}\delta ( \alpha - \alpha_{j})$, with $c_{j}>0$ and $\sum_{j=0}^{m}c_{j}=1$, such that $\vert \int \, \left( P_{\rho_{A}}(\alpha) - \sum_{j=0}^{m}c_{j}\delta(\alpha - \alpha_{j}) \right) f(\alpha){d^{2}\alpha \over \pi}  \vert \le \epsilon$. In terms of quantum states, this implies that the weak limit $\lim_{m \rightarrow \infty}\sum_{j=0}^{m}c_{j}\ket{\alpha_{j}}_{A}\bra{\alpha_{j}} = \int \, {d^{2}\alpha \over \pi} P_{\rho_{A}}(\alpha)\ket{\alpha}_{A} \bra{\alpha} = \rho_{A}$. Because weak convergence and trace norm convergence coincide on the
 set of quantum states (\cite{holevoqubook}, Lemma 11.1), the sequence of classical states $\rho_{A}^{(m)}:=\sum_{j=0}^{m}c_{j}\ket{\alpha_{j}}_{A}\bra{\alpha_{j}}$ converges to $\rho_{A}$ in trace norm. For each $m$, $\rho_{A}^{(m)}$ permits an extension $\rho_{AE}^{(m)} = \sum_{j=0}^{m} c_{j} \ket{\alpha_j}_A\bra{\alpha_j}\otimes \ket{\alpha'_j}_E\bra{\alpha'_j}$, where $\braket{\alpha'_i}{\alpha'_j} = \delta_{ij}$, so the extension is diagonal with respect to an orthogonal basis $\{ \ket{\alpha_j}_A\ket{\alpha'_j}_E \}_{j=0,\ldots ,m}$. Therefore, the $m$-coherence of $\rho_{AE}^{(m)}$ is zero; in fact, $\mathcal{C}_\alpha(\rho_{AE}^{(m)},N) = 0$ for every $N$. By the above construction, the sequence $\rho^{(m)}_{AE}$ satisfies $\lim_{m\rightarrow \infty} \mathrm{Tr}_E(\rho^{(m)}_{AE}) = \lim_{m\rightarrow \infty} \rho_{A}^{(m)} = \rho_A$. By of the contractivity of the trace distance under partial trace and the assumed continuity of the coherence measure $\mathcal{C}$, we have $\mathcal{C}_\alpha(\rho_A) = 0$.

To prove the converse, first suppose that $\mathcal{C}_\alpha(\rho_A) = 0$. By the definition of $\alpha$-coherence, there exists a sequence of extensions $\rho^{(n)}_{AE}$ such that as $n \rightarrow \infty$ and $N \rightarrow \infty$, $\mathcal{C}(\Phi^{(N)}_{GS}(\rho^{(n)}_{AE})) \rightarrow 0$  and $\mathrm{Tr}_{E}(\rho^{(n)}_{AE}) \rightarrow \rho_A$. Therefore, for any $\epsilon' >0$, for sufficiently large $n$, there exists some value $N_{max}$ such that for every $N > N_{max}$, we have  $\mathcal{C}(\Phi^{(N)}_{GS}(\rho^{(n)}_{AE})) < \epsilon'$. Since the $N$-coherence is arbitrarily small for sufficiently large $n$ and $N$, this further implies that for every $\rho^{(n)}_{AE}$, where $n$ is sufficiently large, there exists some state $\sigma^{(n)}_{AE}$ in a small trace norm neighborhood of $\rho^{(n)}_{AE}$ that is incoherent (i.e., $\sigma^{(n)}_{AE} = \sum_{j} c_{j} \ket{\alpha_j}_{A}\bra{\alpha_j} \otimes\ket{\vec{\alpha}'_j}_{E}\bra{\vec{\alpha}'_j}$, with $\frac{1}{2}\Vert \rho^{(n)}_{AE} - \sigma^{(n)}_{AE} \Vert_{tr} < \epsilon''$ for an $\epsilon'' >0$). If we were to choose $n$ such that for every $n > n_{max}$ for some $n_{max}$, $\mathrm{Tr}_E(\rho^n_{AE}) \in \mathcal{B}(\rho_A, \epsilon/2)$, and also choose $\epsilon'' = \epsilon/2$, we will have $\mathrm{Tr}_E(\sigma^n_{AE}) \in \mathcal{B}(\rho_A, \epsilon)$. As $\mathrm{Tr}_{E}\sigma^{(n)}_{AE}$ has a P distribution which is a regular Borel measure on $\mathbb{C}$, i.e., it is classical, and the set of classical states is closed and contains no isolated points \cite{bach}, $\rho_A$ is also classical, which completes the proof.
\end{proof}

In quantum optics, the nonclassical character of a quantum state is usually manifest in the measurement statistics of moments of the quadrature or number operators. Specifically, a classical P distribution constrains these correlation functions to satisfy linear or nonlinear inequalities, depending on the nonclassical features of interest \cite{agarwalcorr,Ryl2017}. Theorem 1 extends the general operational content of the fact that a quantum state associated with a P distribution that is a \textit{bona fide} probability distribution fails to exhibit nonclassical characteristics.  It implies that if a quantum system is described by a state that is indistinguishable from a classical state, then it is impossible to extract any non-classical resource from the system by using the free operations of the coherence resource theory in which the resource is measured by $\mathcal{C}$.

We now consider a possible resource theory where the ``free" operations are linear optical operations, which we define as operations achievable using some combination of linear optical elements such as beam splitters, phase shifters, half and quarter wave plates. Concatenations of these elementary operations forms the most readily available set of operations for manipulating quantum light in the laboratory today. These elements can address both the spatial and polarization degrees of freedom of light. In addition, we also allow for displacement operations as well as additional ``free'' resources in the form of classical ancillas, where classicality means classical P-distributions.

\begin{definition} [Linear optical maps] \label{def::Maps}
A quantum map $\Phi_{L}$ is called a linear optical map/operation if 

$$\Phi_L(\rho_A) = \mathrm{Tr}_E(U_L\rho_A \otimes\sigma_E U_L^\dag)$$

where $U_L$ is some unitary implementable by linear optical elements such as beam splitters, phase shifters, half and quarter wave plates as well as single mode displacement operations. $\sigma_E$ is some classical, possibly multimode ancillary system.

A set of Kraus operators $\{K_i\}$ satisfying $\sum_i K_i^\dag K_i = \openone$ with corresponding POVM elements $K^\dag_i K_i$ representing some classical measurement outcome $i$ is called a linear optical measurement if classical measurement outcomes can be obtained via a linear optical map, i.e. there exists $U_L$ and classical ancilla $\sigma_E$ and some set of orthogonal vectors $\{\ket{\alpha_i'}_{E'} \}$ such that

$$
\mathrm{Tr}_E( U_L\rho_A \otimes \sigma_{EE'} U_L^\dag )= \sum_i p_i \rho^i_{A} \otimes \ket{\alpha'_i}_{E'}\bra{\alpha'_i}
$$

for some density matrices $\rho^i_{A}$, where $p_i \rho^i_{A} \coloneqq K_i\rho_AK^\dag_i $ and $p_i \coloneqq \mathrm{Tr}(K_i\rho_AK^\dag_i)$

\end{definition}

Before we proceed further, we first make the following observation which will prove useful in the subsequent proofs:

\begin{proposition} \label{prop::classicalAncilla}

For any quantum state $\rho_A$ and any classical quantum state $\sigma_B$

$$\mathcal{C}_{\alpha}(\rho_A \otimes \sigma_B) = \mathcal{C}_\alpha(\rho_A)$$

\end{proposition}

\begin{proof}

Suppose $\mathcal{C}_{\alpha}(\rho_A) = C$. This implies that that there exists some sequence of extensions such that $\rho^n_{AE}$ satisfying $\lim_{n\rightarrow \infty} \mathrm{Tr}(\rho^n_{AE}) = \rho_A$, such that for any $\epsilon  > 0 $, there exists sufficiently large $n$ and $N$ such that $|\mathcal{C} ( \Phi^{(N)}_{GS}(\rho_{AE}) )  - C | \leq \epsilon$. If $\sigma_B$ is a classical state, then there exists some sequence of states $\sigma^m_{BE'}$ satisfying $\lim_{m \rightarrow \infty} \mathrm{Tr}(\sigma^m_{BE'}) = \sigma_B$ such that for sufficiently large $m$ and every $M$, $\mathcal{C} ( \Phi^{(M)}_{GS}(\sigma^m_{BE'}) )  = 0$, so $\sigma^m_{BE'} = \sum_i c(i) \ket{\alpha_i}_{B}\bra{\alpha_i} \otimes\ket{\vec{\alpha}'_i}_{E'}\bra{\vec{\alpha}'_i}$ and $\braket{\alpha_i}{\alpha_j}_B\braket{\vec{\alpha}'_i}{\vec{\alpha}'_j}_{E'} = \delta_{ij}$. As a result, we have $ \mathcal{C}(\Phi^{(N)}_{GS}(\rho^n_{AE} \otimes \sigma^n_{BE'}))  =  \mathcal{C} ( \Phi^{(N)}_{GS} (\sum_i c(i) \rho^n_{AE}  \otimes \ket{\alpha_i}_{B}\bra{\alpha_i}\otimes  \ket{\vec{\alpha}'_i}_{E'}\bra{\vec{\alpha}'_i}))  =   \sum_{i} c(i) \mathcal{C} ( \Phi^{(N)}_{GS}(\rho^n_{AE}) ) =  \mathcal{C} ( \Phi^{(N)}_{GS}(\rho^n_{AE}) ) $ , where the second inequality comes about because $\rho_{AE} \otimes \ket{\alpha_i}_{B}\bra{\alpha_i}\otimes  \ket{\vec{\alpha}_i}_{E'}\bra{\vec{\alpha}_i}$ occupies orthogonal subspaces for different values of $i$. Therefore, there exists a sequence such that $\lim_{n \rightarrow \infty} \mathrm{Tr}_{EE'}(\rho^n_{AE} \otimes \sigma^n_{BE'}) = \rho_A \otimes \sigma_B$ and $\mathcal{C}(\Phi^{(N)}_{GS}(\rho^n_{AE} \otimes \sigma^n_{BE'})) = \mathcal{C} ( \Phi^{(N)}_{GS}(\rho^n_{AE}) )$, which implies $\mathcal{C}_{\alpha}(\rho_A \otimes \sigma_B) = \mathcal{C}_{\alpha}(\rho_A)$ if $\sigma_B$ is classical.

\end{proof}

We can then consider a resource theory based on the non-classical P distributions and linear optical operations. 
\begin{definition} \label{def::monotonicity}

We call $\mathcal{Q}$ a non-classicality measure if the following conditions are satisfied:

\begin{enumerate}
\item $\mathcal{Q}(\rho) = 0$ iff $\rho$ is classical.

\item  

	\begin{enumerate}
	
	\item  (Weak monotonicity) $\mathcal{Q}$ is monotonically decreasing under linear optical operations $\Phi_{L}$, 	i.e. $\mathcal{Q}(\rho_A)\geq \mathcal{Q}(\Phi_L (\rho_A))$.
	
	\item (Strong monotonicity) Let $\{ K_i \}$ be a set of Kraus operators corresponding to a linear optical measurement with outcomes $i$. Then $\mathcal{Q}$ is non-increasing when averaged over measurement outcomes $i$, i.e. $\mathcal{Q}(\rho) \geq \sum_i p_i\mathcal{Q}(\rho_i)$ where $p_i \coloneqq \mathrm{Tr}(K^\dag_iK_i \rho )$ and $\rho_i \coloneqq \frac{1}{p_i} K_i \rho K^\dag_i$.
	
	\end{enumerate}

\item $\mathcal{Q}$ is convex, i.e. $\mathcal{Q}(\sum_i p_i \rho_i) \leq \sum_i p_i \mathcal{Q}( \rho_i)$
\end{enumerate}

\end{definition}

We have the following result:

\begin{theorem}
The $\alpha$-coherence is a non-classicality measure.
\end{theorem}

\begin{proof}
The first condition is already proven.

We now prove weak monotonicity. Consider the state $\rho_A$. Suppose $\mathcal{C}_{\alpha}(\rho_A) = C$. This implies that that there exists some sequence of extensions such that $\rho^n_{AE}$ satisfying $\lim_{n\rightarrow \infty} \mathrm{Tr}(\rho^n_{AE}) = \rho_A$, such that for any $\epsilon  > 0 $, there exists sufficiently large $n$ and $N$ such that $|\mathcal{C} ( \Phi^{(N)}_{GS}(\rho_{AE}) )  - C | \leq \epsilon$. Consider the linear map $\Phi_L(\rho_A) = \mathrm{Tr}_{E'}(U_L\rho_A \otimes \sigma_{E'} U_L^\dag)$. By definition, $\sigma_{E'}$ has a classical P distribution, so there exists some sequence of states satisfying $\lim_{m \rightarrow \infty} \mathrm{Tr}(\sigma^m_{E'E''}) = \sigma_{E'}$ such that for sufficiently large $m$ and every $M$, $\mathcal{C} ( \Phi^{(M)}_{GS}(\sigma^m_{E'E''}) )  = 0$, so $\sigma^m_{E'E''} = \sum_i c(i) \ket{\alpha_i}_{E'}\bra{\alpha_i} \otimes\ket{\vec{\alpha}'_i}_{E''}\bra{\vec{\alpha}'_i}$ and $\braket{\alpha_i}{\alpha_j}_B\braket{\vec{\alpha}'_i}{\vec{\alpha}'_j}_{E'} = \delta_{ij}$. From Proposition~\ref{prop::classicalAncilla}, we know that the sequence $\rho^n_{AE}\otimes \sigma^n_{E'E''}$ satisfies $\mathcal{C}(\Phi^{(N)}_{GS}(\rho^n_{AE} \otimes \sigma^n_{E'E''})) = \mathcal{C} ( \Phi^{(N)}_{GS}(\rho^n_{AE}) )$ and $\lim_{n\rightarrow \infty}\mathrm{Tr}_{EE'E''}(\rho^n_{AE}\otimes \sigma^n_{E'E''}) = \rho_A $. We also note that the unitary operation does not change the coherence, so that $\mathcal{C}(\Phi^{(N)}_{GS}(U_L\rho^n_{AE}\otimes \sigma^n_{E'E''} U_L^\dag)) = \mathcal{C}(\Phi'^{(N)}_{GS}(\rho^n_{AE}\otimes \sigma^n_{E'E''} ))$. This is because unitary linear operations always map products of coherent states to another product of coherent states $U_L \ket{\vec{\alpha}} = \ket{\vec{\alpha}'}$. As a consequence the linear operations simply transforms the CNOT type operations to another CNOT type unitary: $U^\dag_L U_{\vec{\alpha}}U_L = U_{\vec{\alpha'}}$. Since the sequence $U_L\rho^n_{AE}\otimes \sigma^n_{E'E''} U_L^\dag$ is in general, a suboptimal sequence of states satisfying $\lim_{n\rightarrow \infty}\mathrm{Tr}_{EE'E''}(U_L \rho^n_{AE}\otimes \sigma^n_{E'E''} U_L^\dag) = \Phi_L(\rho_A)$, we must have $\mathcal{C}_{\alpha}(\rho_A)\geq \mathcal{C}_{\alpha}(\Phi_L (\rho_A))$, which proves weak monotonicity.


The proof of strong monotonicity proceeds similarly. Following from the argument for weak monotonicity, we suppose the linear optical measurement is implemented via the map $\Phi_L(\rho_A \otimes \sigma_{E'E''}) = \mathrm{Tr}_{E'}(U_L\rho_A \otimes \sigma_{E'E''} U_L^\dag) = \sum_i \tau^i_{A} \otimes \ket{\alpha'_i}_{E''}\bra{\alpha'_i}$ and $\tau^i_{A} = K_i\rho_AK^\dag_i$, where $\{ \ket{\alpha'_i}_{E''} \}$ is an orthogonal set. As before,  consider the sequence of extensions such that $\rho^n_{AE}$ satisfying $\lim_{n\rightarrow \infty} \mathrm{Tr}(\rho^n_{AE}) = \rho_A$, such that for any $\epsilon  > 0 $, there exists sufficiently large $n$ and $N$ such that $|\mathcal{C} ( \Phi^{(N)}_{GS}(\rho^n_{AE}) )  - C | \leq \epsilon$. Since $\sigma_{E'E''}$ is classical, there exists some sequence of states satisfying $\lim_{m \rightarrow \infty} \mathrm{Tr}_{E'''}(\sigma^m_{E'E''E'''}) = \sigma_{E'E''}$ such that for sufficiently large $m$ and every $M$, $\mathcal{C} ( \Phi^{(M)}_{GS}(\sigma^m_{E'E''E'''}) )  = 0$, so $\sigma^m_{E'E''E'''} = \sum_i c(i) \ket{\vec{\alpha}'_i}_{E'E''E'''}\bra{\vec{\alpha}'_i} $ and $\braket{\vec{\alpha}'_i}{\vec{\alpha}'_j}_{E'E''E'''} = \delta_{ij}$. 

The above unitary does not change the coherence, so $ \mathcal{C}(\Phi^{(N)'}_{GS}(\rho^n_{AE}\otimes \sigma^n_{E'E''E'''} ) )= \mathcal{C}(\Phi^{(N)}_{GS}(U_L \rho^n_{AE}\otimes \sigma^n_{E'E''E'''} U^\dag_L ) )$.  Consider now the sequence of states $\tau^n_{AE'E''E'''} \coloneqq U_L \rho^n_{AE}\otimes \sigma^n_{E'E''E'''} U^\dag_L  $ and $\tau^n_{AE''E'''} \coloneqq \mathrm{Tr}_{E'} (U_L \rho^n_{AE}\otimes \sigma^n_{E'E''E'''} U^\dag_L )$. From the definition of a linear optical measurement, the subsystem $E''$ stores classical orthogonal measurement outcomes, so $\tau^n_{AE''E'''} = \sum_i \tau^{n,i}_{AE'''} \otimes \ket{\alpha'_i}_{E''}\bra{\alpha'_i} $. Observe that $\lim_{n\rightarrow \infty} \mathrm{Tr}_{E'''}(\tau^{n,i}_{AE'''}) = K_i\rho_A K^\dag_i = p_i \rho^i_A$, so $\tau^n_{AE''E'''}$ and hence $\tau^n_{AE'E''E'''}$ are in fact a sequences of extensions approaching the state $\sum_i p_i \rho^i_A \otimes \ket{\alpha'_i}_{E''}\bra{\alpha'_i}$. Since this sequence of states is not necessarily optimal, we have that $ \mathcal{C}_{\alpha}(\rho_A)  \geq   \mathcal{C}_{\alpha}( \sum_i p_i \rho_A^i \otimes \ket{\alpha'_i}_{E''}\bra{\alpha'_i}) =  \sum_i p_i \mathcal{C}_{\alpha}(  \rho_A^i \otimes \ket{\alpha'_i}_{E''}\bra{\alpha'_i}) = \sum_i p_i \mathcal{C}_{\alpha}(  \rho_A^i )$. In the last equality, we used the fact that $\ket{\alpha'_i}$ specifies orthogonal subspaces for different $i$. This is sufficient to prove strong monotonicity.


The only thing that remains to be proven is convexity. Let $\rho^{n,i}_{AE}$ be the sequence of extensions of the set of states $\rho^i_A$ satisfying $\lim_{n\rightarrow \infty} \mathrm{Tr}(\rho^{n,i}_{AE}) = \rho^i_A$, such that for any $\epsilon  > 0 $, there exists sufficiently large $n$ and $N$ such that $|\mathcal{C} ( \Phi^{(N)}_{GS}(\rho^{n,i}_{AE}) )  - \mathcal{C}_\alpha (\rho^i_A) | < \epsilon$ for every $i$. It is clear that the state $\mathcal{C}(\Phi^{(N)}_{GS}(\sum_i p_i \rho^{n,i}_{AE} \otimes \ket{\alpha'_{n,i}}_{E'}\bra{\alpha'_{n,i}})) = \sum_i p_i \mathcal{C}(\Phi^{(N)}_{GS}(\rho_A^{n,i}))$ when $\{ \ket{\alpha'_i}_{E'} \}$ is an orthogonal set, which implies $\sum_i p_i \mathcal{C}_\alpha(\rho_A^i)\geq \mathcal{C}_{\alpha}(\sum_i p_i \rho^i_{A})$ since the sequence $\rho^{n,i}_{AE} \otimes \ket{\alpha'_{n,i}}_{E'}\bra{\alpha'_{n,i}}$ may be suboptimal for $\mathcal{C}_{\alpha}(\sum_i p_i \rho^i_{A})$. This proves convexity.

%
\end{proof}

\section{Examples}

Here, we present some numerical plots of the $\alpha$-coherence for some important classes of pure states. For pure states in particular, the optimization is much simpler as the only possible extensions are trivial, thus allowing us to sidestep part of the optimization involved in Definition~\ref{def::alphaCoherence}. For the examples considered, we will employ the relative entropy of coherence~\cite{Baumgratz14} as our coherence measure.

In Fig~\ref{fig::compare1} we see a comparison of the $\alpha$-coherence for the even and odd cat states $\ket{\alpha} \pm \ket{-\alpha}$, Fock states $\ket{n}$, and squeezed states $S(\xi) \ket{0}$ with a real squeezing parameter $\xi$. We see that for both Fock states and squeezed states, the $\alpha$-coherence monotonically increases, indicating strong nonclassicality as is expected. In the case of odd cat states, we see strong nonclassicality in the region where $\alpha \approx 0$. This is because in the limit $\alpha \rightarrow 0^+$, we know that the odd cat approaches the single photon state, an archetypical example of nonclassical light. In contrast, for the even cat states, as $\alpha \rightarrow 0^+$, the state approaches the vacuum, which is classical, and this is reflected in a vanishing $\alpha$-coherence. It is interesting to note that non-classicality peaks most strongly in the region $\alpha = 1$. We interpret this as a signature of the infinite dimensional nature of the underlying Hilbert space, as the state tends towards a 2 dimensional superposition as $\alpha \rightarrow \infty$. We also note that in the limit $\alpha \rightarrow \infty$, the $\alpha$-coherence asymptotically tends towards a constant value, in contrast to a macroscopicity measure~\cite{Lee2011} which will increase with the separation $\alpha$.

We also consider things from the point of view of efficiency. Fig~\ref{fig::compare2} is a numerical plot of the nonclassicality for a given mean particle number. We see that the Fock states are the most nonclassical states on a per particle basis over the region considered, which is again not unexpected due to the granular nature of this form of light.

\begin{figure}[t]
\includegraphics[width=0.3\linewidth]{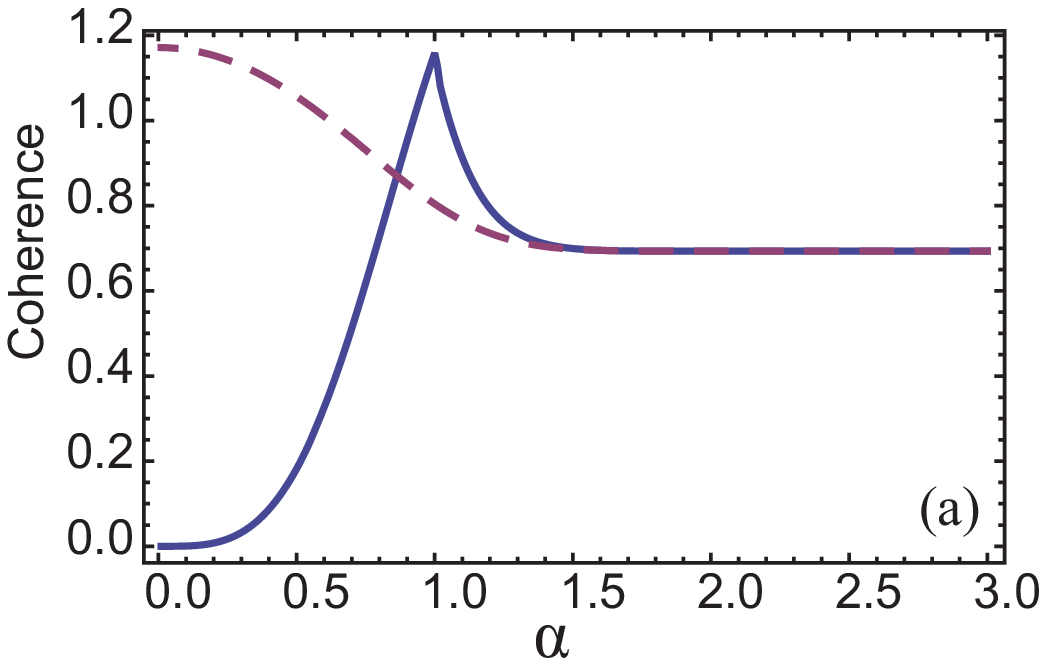}
\includegraphics[width=0.3\linewidth]{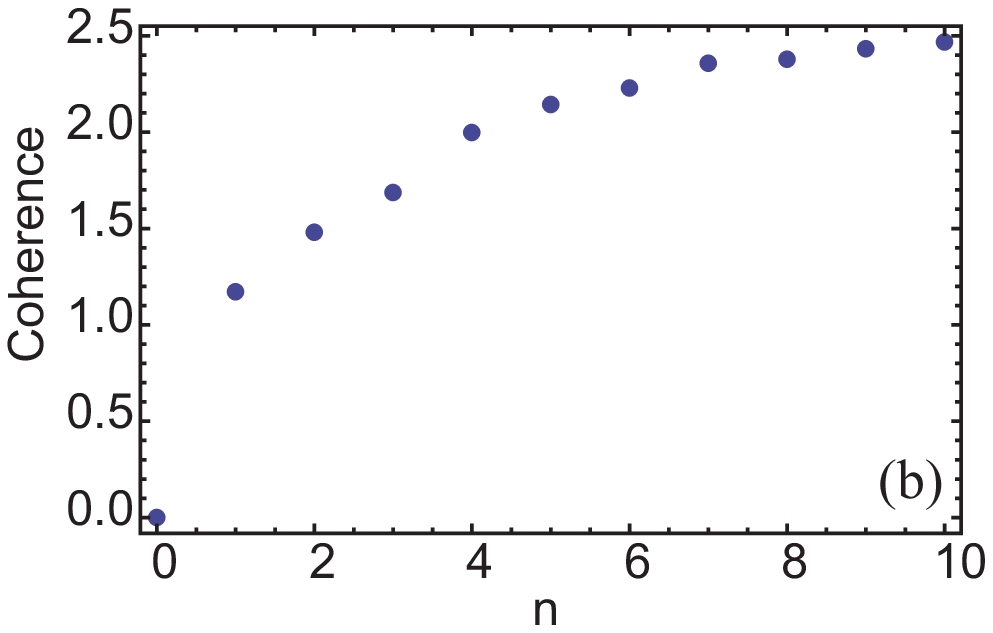}
\includegraphics[width=0.31\linewidth]{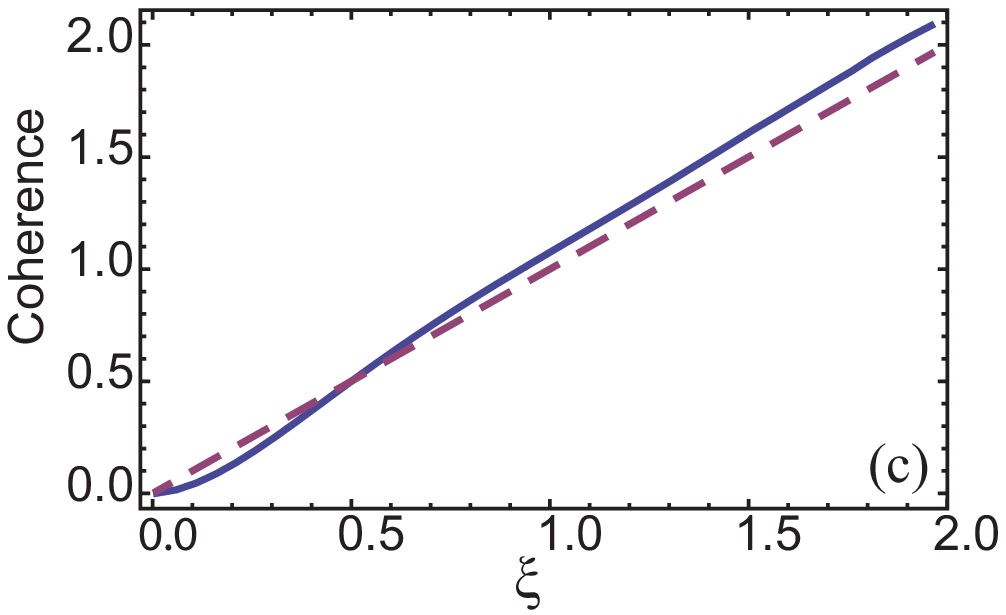}
\caption{Coherence measure ${\cal C}(\rho) = S(\rho_{\rm diag}) - S(\rho)$ for photonic states. (a) Even(solid line) and odd(dotted line) cat states $\ket{\alpha} \pm \ket{-\alpha}$, (b) Fock states $\ket{n}$, and (c) squeezed states $S(\xi) \ket{0}$ are compared. For cat states, the degree of coherence approaches to $\log 2$, which is the maximum coherence for qubit states when $\alpha$ approaches infinity.
Degrees of coherence for Fock states and squeezed states increase as a photon number $n$ and squeezing parameter $\xi$ increase.
}
\label{fig::compare1}
\end{figure}

\begin{figure}[t] 
\includegraphics[width=0.6\linewidth]{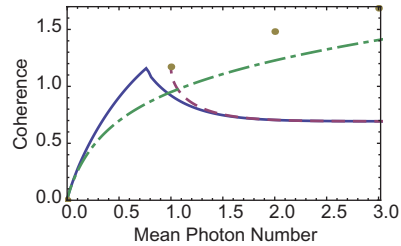}
\caption{Coherence measure  ${\cal C}(\rho) = S(\rho_{\rm diag}) - S(\rho)$ for photonic states. Even(solid line) and odd(dotted line) cat states $\ket{\alpha} \pm \ket{-\alpha}$, Fock states(circular points) $\ket{n}$, and squeezed states(dashed-dotted line) $S(\xi) \ket{0}$ are compared.
Coherence of the states are plotted for given mean photon numbers $\langle n \rangle = \langle a^\dagger a \rangle$.
 For cat states, the degree of coherence approaches to $\log 2$, which is the maximum coherence for qubit states when $\alpha$ approaches infinity.
Degrees of coherence for Fock states and squeezed states increase as a photon number $n$ and squeezing parameter $\xi$ increase, respectively.
}

\label{fig::compare2}
\end{figure}

\section{Other possible measures}

Here, we consider another possible measure of non-classicality based on the negative volume of the P distribution. In the most general case, negativities in the $P$ distribution can come in the form of regular continuous functions, which are directly accessible, or singularities. Suppose we restrict ourselves to the case where the $P$ distribution is a regular continuous function. We can then consider the following:

\begin{definition} [Negativity]
Let the P distribution of the of the state $\rho$ be given by $p(\alpha)$ where $p(\alpha)$ is a regular function. Let $\mathcal{N} = \{ \alpha \mid p(\alpha) \ngeq 0\}$, then then the quantity:

$$\mathcal{C}_- (\rho) = - \int_\mathcal{N}d^2\alpha \, p(\alpha)$$

is called the negativity of the P distribution.

\end{definition}

The following result shows that both the $\alpha$-coherence and the negativity of the P distribution belong to similar resource theories, which further supports the argument that the $\alpha$-coherence is closely related to negativities in the P distribution.

\begin{theorem}
Suppose for some $\rho$, $\mathcal{C}_- (\rho)$ is finite integrable. Then $\mathcal{C}_-$ is a non-classicality measure with respect to the set of states with positive P distributions.
\end{theorem}

\begin{proof}
It is obvious that if $p(\alpha)\geq 0 $ for all $\alpha$, then the P distribution is classical and $\mathcal{C}_- (\rho) = 0$. So the first condition is automatically satisfied.

For weak monotonicity, observe that a linear map always maps a state with classical P distribution to another state with classical P distribution. Therefore, we must have $\Phi_L(\ket{\alpha}\bra{\alpha})= \int d^2 \alpha' \; r_\alpha(\alpha')\ket{\alpha'}\bra{\alpha'}$ where $r_\alpha(\alpha')$ is some classical P distribution. Let $\rho = \int d^2\alpha\; p(\alpha)\ketbra{\alpha}{\alpha}$ where $ p(\alpha)d^2\alpha\ $ is a signed measure with density $p(\alpha)$ (the existence of such an expression follows from the fact that the set of finite linear combinations $\sum_{j=0}^{m}c_{j}\ket{\alpha_{j}}\bra{\alpha_{j}}$ is trace norm dense in the set of quantum states \cite{daviesbook}).

By the Hahn decomposition theorem, we define the positive subset $\mathcal{P} = \{ \alpha \mid p(\alpha) \geq 0 \}$ and the negative subset as $\mathcal{N} = \{ \alpha \mid p(\alpha) \le 0 \}$, so $\rho = \left ( \int_\mathcal{P} + \int_\mathcal{N} \right ) d^2\alpha\; p(\alpha)\ketbra{\alpha}{\alpha}$ and $\Phi_L(\rho) =  \left ( \int_\mathcal{P} + \int_\mathcal{N} \right ) d^2\alpha\; p(\alpha) \Phi_L(\ketbra{\alpha}{\alpha}) = \left ( \int_\mathcal{P} + \int_\mathcal{N} \right ) d^2\alpha\; p(\alpha)\int d^2 \alpha' \; r_\alpha(\alpha')\ket{\alpha'}\bra{\alpha'}$. Since $r_\alpha(\alpha')$ is classical and hence always non-negative, we must have that $\mathcal{C}_-(\Phi_L(\rho)) \leq -\int_\mathcal{N}  d^2\alpha\; \int d^2 \alpha' \; p(\alpha) r_\alpha(\alpha') = -\int_\mathcal{N} d^2 \alpha' \; p(\alpha) = \mathcal{C}_-(\rho)$, which proves weak monotonicity.

We now prove strong monotonicity. Recall the definition of a linear optical measurement. Let $\rho_A = \int d^2\alpha\; p(\alpha)\ketbra{\alpha}{\alpha}$ where $p(\alpha)$ may be nonclassical. Any linear optical measurement may be performed by a linear optical unitary operation with classical ancilla:  $ U_L \rho_A \otimes \sigma_{EE'} U^\dag_L  = U_L \int d^2 \alpha \; p(\alpha) \ket{\alpha}_A \bra{\alpha} \otimes \int d^{2M} \alpha' r(\vec{\alpha}') \ket{\vec{\alpha}'}_{EE'} \bra{\vec{\alpha}'}  U^\dag_L = \int d^2 \alpha d^{2M}\vec{\alpha}' p(\alpha)r(\vec{\alpha}') U_L\ket{\alpha, \vec{\alpha}'}_{AEE'}\bra{\alpha, \vec{\alpha'}}U^\dag_L$. $r(\vec{\alpha}')$ is a classical non-negative distribution over $M$ modes. Since $U_L$ is a linear optical unitary, it always maps a product of coherent states to another product of coherent states, so we can write $U_L \ket{\alpha, \vec{\alpha}'}_{AEE'} = \ket{\beta(\alpha, \vec{\alpha}'), \vec{\beta}'(\alpha, \vec{\alpha}')}_{AEE'}$. Since $U_L$ implements a linear optical measurement, there must exist projectors $\ket{\gamma_i}_{E'}\bra{\gamma_i}$ on the subsystem $E'$ such that $p_i \rho_i = \mathrm{Tr}_{EE'} (U_L \rho_A \otimes \sigma_{EE'} U^\dag_L \openone_E\otimes\ket{\gamma_i}_{E'}\bra{\gamma_i})$. Expanding this expression, we get $p_i \rho_i = \int d^2 \alpha d^{2M}\vec{\alpha}' p(\alpha)r(\vec{\alpha}') \bra{\vec{\beta}'(\alpha, \vec{\alpha}')}(\openone_E\otimes\ket{\gamma_i}_{E'}\bra{\gamma_i}) \ket{\vec{\beta}'(\alpha, \vec{\alpha}')}_{EE'} \ket{\beta(\alpha, \vec{\alpha}')}_{A}\bra{\beta(\alpha, \vec{\alpha}')} $. We see that since the terms $r(\vec{\alpha}')$ and $\bra{\vec{\beta}'(\alpha, \vec{\alpha}')}(\openone_E\otimes\ket{\gamma_i}_{E'}\bra{\gamma_i}) \ket{\vec{\beta}'(\alpha, \vec{\alpha}')}_{EE'}$ in the integral are both non-negative, we can upper bound the negativity of $\rho_i$ by simply integrating over the entire negative subset of $p(\alpha)$, regardless of the measurement outcomes $i$. As a result, we can write 

\begin{align*}
\sum_i p_i & \mathcal{C}_- ( \rho_i) \\
&\leq -\int_\mathcal{N} d^2 \alpha d^{2M}\vec{\alpha}' p(\alpha)r(\vec{\alpha}') \times \\ &\;\bra{\vec{\beta}'(\alpha, \vec{\alpha}')}(\openone_E\otimes\sum
_i \ket{\gamma_i}_{E'}\bra{\gamma_i}) \ket{\vec{\beta}'(\alpha, \vec{\alpha}')}_{EE'} \\
&\leq -\int_\mathcal{N} d^2 \alpha d^{2M}\vec{\alpha}' p(\alpha)r(\vec{\alpha}') = \mathcal{C}_-(\rho)
\end{align*} 

where the last inequality is because the sum $\bra{\vec{\beta}'(\alpha, \vec{\alpha}')}(\openone_E\otimes \sum_i \ket{\gamma_i}_{E'}\bra{\gamma_i}) \ket{\vec{\beta}'(\alpha, \vec{\alpha}')}_{EE'}$ is a sum of probability outcomes. This is proves strong monotonicity.

Convexity is also guaranteed. Let $p(\alpha)$ and $q(\alpha)$ be the P distributions of $\rho$ and $\sigma$ respectively. The P distribution of the mixture $r\rho + (1-r) \sigma$ is $rp(\alpha) + (1-r)q(\alpha)$. Since $-\int_\mathcal{N} d^2\alpha (rp(\alpha) + (1-r)q(\alpha)) = -r\int_\mathcal{N} d^2\alpha \, q(\alpha) - (1-r)\int_\mathcal{N} d^2\alpha \,q(\alpha) \leq r \mathcal{C}_- (\rho)+ (1-r)\mathcal{C}_- (\sigma)$, $\mathcal{C}_-$ must be convex. The inequality occurs because the the largest negative set $\mathcal{N}$ for for mixture may be suboptimal for the individual states $\rho$ and $\sigma$.

\end{proof}

Still other possible measures of nonclassicality measures can also be constructed. For instance, we can also consider geometric measures of nonclassicality. Suppose we have some distance measure $D(\rho,\sigma)$ over the Hilbert space that is monotonically decreasing under quantum operations over both its arguments. Then it is immediately clear that the quantity $\inf_{\sigma \in \mathcal{P}^+} D(\rho,\sigma)$, where the optimization is over all classical states, will satisfy at least the weak monotonicity condition laid out in Definition~\ref{def::monotonicity}.

\section{Conclusion.}

In this paper, we described a general procedure that allows us to quantify the superposition amongst any complete set of quantum states, whether they are orthogonal or not. The key insight here is that the scheme laid out by by Baumgratz {\it et al.}~\cite{Baumgratz14} can be generalized via a reasonably motivated orthogonalization procedure. This orthogonalization procedure is then applied to the set of coherent states as a special case and the resulting coherence measure, the $\alpha$-coherence, is shown to identify incoherent states with nonclassical states  in the sense of the Glauber-Sudarshan $P$ distribution. This demonstrates that states with nonclassical $P$ distributions are essentially the limiting case of the same quantum resources identified in~\cite{Baumgratz14}, when the incoherent basis is chosen as the set of coherent states. The $\alpha$-coherence also belongs to a class of resource theoretic nonclassicality measures that we refer to as a linear optical resource theory. This strongly implies that linear optical monotones are appropriate measures of the nonclassicality of light. The results also suggest possible deeper connections between incoherent operations and linear optical elements that opens up potentially interesting new lines of investigation. 

\section{Acknowledgements}
This work was supported by the National Research Foundation of Korea (NRF) through a grant funded by the Korea government (MSIP) (Grant No. 2010-0018295). K.C. Tan and T. Volkoff was supported by Korea Research Fellowship Program through the National Research Foundation of Korea (NRF) funded by the Ministry of Science, ICT and Future Planning (Grant No. 2016H1D3A1938100 and 2016H1D3A1908876). H. Kwon was supported by the Global Ph.D. Fellowship Program through the NRF funded by the Ministry of Education (Grant No. 2012-003435).

\end{document}